\documentclass[conference]{IEEEtran}
\AtBeginDocument{}
\usepackage{calrsfs}
\DeclareMathAlphabet{\pazocal}{OMS}{zplm}{m}{n}
\usepackage{epsfig,endnotes,xspace,url}
\usepackage{amsmath}
\usepackage{amssymb}
\usepackage{enumerate}
\usepackage{varioref}
\usepackage{graphicx}
\usepackage{epstopdf}
\usepackage{xcolor}
\usepackage{multirow}
\usepackage{subfigure}
\usepackage{comment}
\usepackage{multirow}
\usepackage{color}
\usepackage{wrapfig}
\usepackage{cite}
\renewcommand{\baselinestretch}{1.2}

\usepackage{algorithm}
\usepackage[noend]{algpseudocode}

\algdef{SE}[DOWHILE]{Do}{doWhile}{\algorithmicdo}[1]{\algorithmicwhile\ #1}%
\newtheorem{theorem}{Theorem}
\newtheorem{principle}{Principle}
\newtheorem{proposition}{Proposition}
\newtheorem{lemma}[theorem]{Lemma}
\newtheorem{definition}{Definition}

\newtheorem{obersvation}{Observation}

\newcommand{\perturb}{\ensuremath{\pi}\xspace}
\newcommand{\aggregate}{\ensuremath{\Gamma}\xspace}

\newcommand{\mypara}[1]{\noindent\textbf{#1 }}

\newcommand{\tuple}[1]{\ensuremath{\langle #1 \rangle}\xspace}

\newcommand{\Domain}{\ensuremath{D}\xspace}

\newcommand{\says}[2]{{\color{blue}{#1 says: }{#2}}\xspace}

\newcommand{\rappor}{\texttt{SPM}\xspace}
\newcommand{\spm}{\texttt{SPM}\xspace}

\renewcommand{\Pr}[1]{\ensuremath{\mathsf{Pr}\left[#1\right]}\xspace}

\newcommand{\Pall}[2]{\ensuremath{\mathsf{Piden}_{#2}\left[#1\right]}\xspace}

\newcommand{\rap}{RAPPOR\xspace}
\newcommand{\blh}{BLH\xspace}

\newcommand{\ind}{\mathsf{1}}

\newcommand{\hash}{\texttt{MCM}\xspace}
\newcommand{\product}{\texttt{HieOSM}\xspace}
\newcommand{\itergroup}{\texttt{PEM}\xspace}

\newcommand{\covt}{\texttt{CoverTree}\xspace}

\newcommand{\grr}{GRR\xspace}
\newcommand{\olh}{OLH\xspace}

\newcommand{\oue}{OUE\xspace}

\newcommand{\DE}{DE\xspace}

\newcommand{\PE}{\ensuremath{\mathsf{PE}}\xspace}
\newcommand{\fone}{\ensuremath{\mathsf{F1}}\xspace}
\newcommand{\var}{\ensuremath{\mathsf{Var}}\xspace}
\newcommand{\ncr}{\ensuremath{\mathsf{NCR}}\xspace}

\begin{document}
	\vspace{-1cm}
\author{\IEEEauthorblockN{Tianhao Wang}
\IEEEauthorblockA{Purdue University\\
tianhaowang@purdue.edu}
\and
\IEEEauthorblockN{Ninghui Li}
\IEEEauthorblockA{Purdue University\\
ninghui@cs.purdue.edu}
\and
\IEEEauthorblockN{Somesh Jha}
\IEEEauthorblockA{University of Wisconsin-Madison\\
jha@cs.wisc.edu}}
\vspace{-1cm}
\sloppypar
\newcommand{\fullTitle}{Locally Differentially Private Heavy Hitter Identification}

\title{\fullTitle}

\IEEEoverridecommandlockouts
\makeatletter\def\@IEEEpubidpullup{9\baselineskip}\makeatother
\IEEEpubid{\parbox{\columnwidth}{
	}
	\hspace{\columnsep}\makebox[\columnwidth]{}}

\maketitle
\vspace{-1cm}
\begin{abstract}
The notion of Local Differential Privacy (LDP) enables users to answer sensitive questions while preserving their privacy. The basic LDP frequent oracle protocol enables the aggregator to estimate the frequency of any value. But when the domain of input values is large, finding the most frequent values, also known as the heavy hitters, by estimating the frequencies of all possible values, is computationally infeasible.  
In this paper, we propose an LDP protocol for identifying heavy hitters.  In our proposed protocol, which we call Prefix Extending Method (\itergroup), users are divided into groups, with each group reporting a prefix of her value.  We analyze how to choose optimal parameters for the protocol and identify two design principles for designing LDP protocols with high utility.  Experiments on both synthetic and real-world datasets demonstrate the advantage of our proposed protocol.


\end{abstract}




\section{Introduction}
\label{sec:intro}

In recent years, differential privacy~\cite{Dwo06,DMNS06} has been increasingly
accepted as the \textit{de facto} standard for data privacy in the
research community
\cite{dpbook,Li2016book, abadi2016deep, papernot2016semi, johnson2017practical}.
Recently, techniques for satisfying differential
privacy (DP) in the local setting, which we call {LDP}, have been
deployed.  Such techniques enable gathering of statistics while
preserving privacy of every user, without relying on trust in a single
data curator.  For example, researchers from Google
developed \rap~\cite{rappor,rappor2}, which is included as part of
Chrome.  It enables Google to collect users' answers to questions such as
the default homepage of the browser, the default search engine, and so on,
to understand the unwanted or malicious hijacking of user settings.
Apple~\cite{Apple} also uses similar methods to help with predictions of spelling and other things, but the details of
the algorithm are not public yet.  Samsung proposed a similar
system~\cite{samsung} which enables collection of not only categorical
answers (e.g., screen resolution) but also numerical answers (e.g.,
time of usage, battery volume), although it is not clear whether this
has been deployed by Samsung.

In the LDP setting, each user possesses an input value $v \in \Domain$, and the aggregator wants to learn
the distribution of the input values among all users.  Existing research~\cite{rappor,Bassily2015local,ldpprimitive}
has developed frequency oracle protocols, where the aggregator can estimate the frequency of any chosen value $v \in \Domain$.
When the size of $\Domain$ is small, such frequency oracle protocols can be used to efficiently reconstruct a noisy
approximation of the input distribution.  When the size of $\Domain$ is so large that issuing an oracle query for
each value in it is computationally infeasible, one needs an additional protocol that first identifies a set of
candidate frequent values.  Two protocols for doing this exist~\cite{rappor2,Bassily2015local}.




In this paper, we propose the Prefix Extending Method ($\itergroup$), which is conceptually very simple, and yet is able to provide much better accuracy
than existing protocols, and the advantage is more pronounced as the size of $\Domain$ gets larger.
The basic idea of $\itergroup$ is to gradually identifying longer and longer frequent prefixes. 
For example, if we view \Domain as consisting of length-$m$ binary strings, we divide the users into $g$ groups, where users in each group report 
a prefix of a certain length.  Users in the $j+1$'th group report prefixes of length $\eta$ longer than the $j$'th group, and the $g$'th group
report the whole string.  
Thus the population is divided into $g$ groups of roughly the same size. 
The aggregator uses reports from the first group to finds $C_1$, the set of frequent prefixes, and then uses reports
from the second group to find $C_2$, considering candidates that have prefixes in $C_1$.  The aggregator iterates this process until finding the set of frequent values.

An important parameter in this process is the segment length $\eta$.  Larger $\eta$ would mean higher computational cost.  In terms
of utility, larger $\eta$ means fewer groups and more users in each group, which improves utility.  However, larger $\eta$ means more
candidates to consider in each step, which leads to lower accuracy.  We conduct an utility analysis to study the interactions
of these two effects.  The utility analysis enables us to draw to a conclusion that the first effect dominates the second, and thus
larger $\eta$ results in better utility.  Thus the choice of $\eta$ depends on limitation on the computational resources.
Because of the complexity of the problem, we have to make several simplifying approximations in the analysis.  To validate the analysis,
we accompany each step of the analysis with empirical experiment to show that conclusions drawn from the analysis match empirical ones.

With the analysis, we are able to identify two design principles.  The first is, when asking multiple questions, it
is better to partition the users into groups, and having each group answer one question, as opposed to having each user answer all the
questions splitting the privacy budget.  The second is, one should reduce the number of groups as much as possible when designing LDP protocols, 
as larger group size is a key in achieving accuracy.  Some existing protocols violate these principles, and can be improved by following them.
We expect these principles to guide the design of LDP protocols for other problems.

Finally, we demonstrate the effectiveness of \itergroup by conducting experiments with both synthetic and real-world datasets.
Result shows that \itergroup greatly outperforms existing solutions.

To summarize, we make the following contributions:
\begin{itemize}
	\item We provide new solutions for the privacy-preserving heavy hitter problem.
	The protocol is then analyzed and optimized.
	\item We identify two principles that can guide the design of protocols for other LDP problems.
	\item We demonstrate the effectiveness of our solution using real-world and synthetic datasets.
\end{itemize}

\mypara{Roadmap.}
In Section~\ref{sec:back}, we present LDP and describe existing mechanisms. We then go over the problem definition and existing solutions in Section~\ref{sec:prelim}.
Section~\ref{sec:methods} presents our proposed method and analysis.
The analysis is validated in Section~\ref{sec:ver}.
Experiment results are given in~\ref{sec:eval}.
Finally we discuss related work in Section~\ref{sec:related} and
conclude in Section~\ref{sec:conc}.

\section{Background}
\label{sec:back}

We consider a setting where there are many \emph{users} and one \emph{aggregator}.  Each user possesses an input value $v \in \Domain$, and the aggregator wants to learn (and use) the distribution of the input values among all users, in a way that protects the privacy of individual users.

In the standard (or centralized) setting, each user sends $v$ to the aggregator, which obtains a histogram for the distribution, and can add noises to the histogram to satisfy differential privacy, so that each individual user's input has a limited impact on the output.  In this setting, the aggregator sees the raw input from all users and is trusted to handle these private data correctly.

\subsection{Differential Privacy in Local Setting}

In the local (or distributed) setting, we want to remove the need to trust the aggregator.  To achieve this, each user perturbs the input value $v$ using an algorithm $\perturb$ and sends $\perturb(v)$ to the aggregator.  The formal privacy requirement is that the algorithm $\perturb(\cdot)$ satisfies local differential privacy, defined as follows:

\begin{definition}[Local Differential Privacy] \label{def:dlp}
	An algorithm $\perturb$ satisfies $\epsilon$-local differential privacy ($\epsilon$-LDP), where $\epsilon \geq 0$,
	if and only if for any input $v_1,v_2 \in \Domain$, we have
	\begin{equation*}
	\forall{T\subseteq\! \mathit{Range}(\perturb)}:\; \Pr{\perturb(v_1)\in T} \leq e^{\epsilon}\, \Pr{\perturb(v_2)\in T},
	\end{equation*}
	where $\mathit{Range}(\perturb)$ denotes the set of all possible outputs of the algorithm $\perturb$.
\end{definition}


For an algorithm $\perturb(\cdot)$ to satisfy $\epsilon$-LDP, it must be randomized.   Compared to the centralized setting, the local version of DP offers a stronger level of protection, because each user only reports the perturbed data. Each user's privacy is still protected even if the aggregator is malicious.

\subsection{Frequency Oracles}
\label{subsec:estimators}



A protocol that satisfies LDP is specified by two algorithms: $\perturb$, which is used by each user to perturb her input value, and $\aggregate$, which takes as input the reports from all users, and outputs the desired information.  A basic protocol under LDP is to estimate the frequency of any given value $v\in \Domain$.  In which a protocol, $\aggregate$ outputs an oracle that can be queried for the frequency of each value.  We thus call such a protocol a frequency oracle protocol.

We assume that there are $n$ users, and user $j$'s value is $v^j \in \Domain$, and the domain size is $|\Domain|=d$.

\subsubsection{Generalized Randomized Response (GRR)}

One frequency oracle protocol generalizes the \emph{randomized response} technique~\cite{Warner65}.
In this protocol, $\perturb_{GRR}(v)$ outputs the value $v$ with probability $p=\frac{e^\epsilon}{e^\epsilon+d-1}$, and each value $v'\ne v$ with probability $\frac{1-p}{d-1}=\frac{1}{e^\epsilon+d-1}=\frac{p}{e^\epsilon}$.  In the special case where the value is one bit, i.e., when $d=2$, $\perturb_{GRR}(v)$ keeps the bit unchanged with probability $\frac{e^\epsilon}{e^\epsilon+1}$ and flips it with probability $\frac{1}{e^\epsilon+1}$.


The frequency oracle outputted by $\aggregate_{GRR}$ in this protocol works as follows.  To estimate the frequency of $v$, it counts how many times $v$ is reported and obtains $I_v$, and then outputs $\frac{I_v-nq}{p-q}$.  That is, the frequency estimate is a linear transformation of the noisy count $I_v$, in order to account for the effect of randomized response.  In~\cite{ldpprimitive}, it is shown that this is an unbiased estimation of the true count, and the variance for this estimation is
\begin{equation}\label{var_grr}
\frac{d-2+e^\epsilon}{(e^\epsilon-1)^2}\cdot n.
\end{equation}


The accuracy of this protocol deteriorates fast when the domain size $d$ increases.  The larger $d$ is, the lower the probability that a value is preserved.  This is reflected in the fact that the variance of is linear in $d$. For example, when $\epsilon=\ln 49$, with $d=2^{16}$, we have $p=\frac{49}{65584} \approx 0.00075$,
and variance $\frac{65583}{2304}\approx 28.5n$

More sophisticated frequency estimators have been studied before~\cite{rappor,Bassily2015local,ldpprimitive}.  In~\cite{ldpprimitive}, several such protocols are analyzed, optimized, and compared against each other, and it was found that when $d$ is large, the Optimized Local Hashing (\olh) protocol provides the best accuracy while maintaining a low communication cost.  In this paper, we use the \olh protocol as a primitive and describe it below.

\subsubsection{Optimized Local Hashing (\olh)~\cite{ldpprimitive}}
\label{subsubsec:olh}

The Optimized Local Hashing (\olh) protocol deals with a large domain size $d$ by first using a hash function to map an input value into a smaller domain of size $d'$, typically $d'<< d$, and then applies randomized response to the value in the smaller domain.  In this protocol, both the hashing step and the randomization step result in information loss. The choice of the parameter $d'$ is a tradeoff between losing information during the hashing step and losing information during the randomization step.  In~\cite{ldpprimitive}, it is found that the optimal choice of $d'$ is $e^\epsilon+1$.

In \olh, $\perturb_{OLH}(v)=\tuple{H,\perturb_{GRR}(H(v)}$, where $H$ is randomly chosen from a family of hash functions that hash each value in $\Domain$ to $\{1\ldots d'\}$, where $d'=\lceil e^\epsilon + 1 \rceil$, and $\perturb_{GRR}$ is the perturb algorithm used in generalized randomized response, with probability $p=\frac{e^\epsilon}{e^\epsilon+d'-1}$.

The frequency oracle outputted by $\aggregate_{OLH}$ in this protocol works as follows.  Let $\tuple{H^j,y^j}$ be the report from the $j$'th user.
For each value $v\in \Domain$, the oracle first computes $I_v=|\{j\mid H^j(v) = y^j\}|$.  That is, $I_v$ is the number of reports that ``supports'' that the input is $v$.
The oracle then outputs

%

\begin{align}
\label{eqn:est}
\frac{I_v - n/d'}{p-1/d'}.
\end{align}


The variance of this estimation is
\begin{equation}\label{var_olh}
\frac{4e^\epsilon}{(e^\epsilon-1)^2}\cdot n.
\end{equation}
Compared with \eqref{var_grr}, the factor $d-2+e^\epsilon$ is replaced by $4 e^\epsilon$.  This suggests that for smaller $d$, one is better off with \grr; but for large $d$, \olh is better and has a variance that does not depend on $d$.

We point out that using \olh, each invocation of the frequency oracle takes time linear in the population size.  Furthermore, the computations needed for recovering the frequency of one value are independent from those needed for recovering that of another value.

\mypara{The importance of group size.}
One may notice that the above frequency oracles under LDP all have estimation variance that is linear in $n$, which means that the standard deviation of the estimations is linear in $\sqrt{n}$.  This is shared by all protocols under LDP, and is a fundamental accuracy cost one has to pay in order to achieve LDP~\cite{chan2012optimal}.  What this means, however, is that LDP protocols can be useful only when the group size $n$ is large, and LDP protocols are meaningful only for the frequent values.

We now use some concrete numbers to make these points clear.  For example, to recover a value that is possessed by $0.1\%$ of the population, we have the true count being $0.001n$.  Assuming we choose $\epsilon$ such that $e^\epsilon=10$, then using OLH the standard deviation is $\sqrt{\frac{40}{81} n} \approx 0.7 \sqrt{n}$.  If we desire that the true count is at least 3 times the standard deviation, then we require $0.001n \geq 3\times 0.7 \sqrt{n}$, or $n \geq 4,410,000$.  This suggest that with about $4.5$ million users, we can recover meaningful frequencies for values that appear in at least $0.1 \%$ of the population.  Quadrupling the population size would enables us to reduce this $0.1 \%$ sensitivity threshold by half.
While theoretically there are up to $1000$ values with frequencies $0.1 \%$ or higher, in most distributions there are likely no more than a few dozen of such values, because the most frequent values will appear with frequencies far higher than $0.1 \%$ and the total frequencies of infrequent values can also be substantial.




\section{Problem Definition and Existing Methods}
\label{sec:prelim}

Recall that the aggregator wants to know the distribution of the frequent values.
When the data domain $\Domain$ is relatively small, having a frequency oracle protocol suffices, as the aggregator can invoke the frequency oracle for all values in $\Domain$, and identify the frequent ones.  However, in many applications, the data domain $\Domain$ is very large, e.g., $2^{128}$ when the input values have 16 bytes.  Enumerating through all values in them is computationally infeasible.

In this paper we focus on the problem of identifying frequent values under the LDP setting when the input domain is large.  For simplicity, we assume that each value is represented by a binary string of length $m$, although our method can be easily changed to support more complicated structure of values, such as a value consisting of multiple components.  

\subsection{Problem Definition}

The problem of finding frequent values (heavy hitters) can be defined either as identifying the top-$k$ values or finding values that appear above a certain threshold.  We assume that each user has a single value, and thus each frequency threshold can be approximately translated into a $k$ value.  Also, note that when the population size $n$ and the privacy budget $\epsilon$ is set, the number of threshold above which one can estimate frequencies accurately is more or less fixed.  We use the top-$k$ version of definition.


\begin{definition}[Top-$k$ Heavy Hitter.]
	\label{def:topk}
	Given a multi-set $\{\{v^1,v^2\ldots, v^n\}\}\in D^n$. An element $x\in \Domain$ is a top-$k$ heavy hitter if its frequency $f_x=\frac{|\{j|j\in[n]\wedge v^j=x\}|}{n}$ is ranked among top $k$ frequencies of all possible values.
\end{definition}

Suppose that each user has a length $m=128$ binary string $v$ as input value, the naive approach of querying the frequency of each string requires $2^{128}$ oracle queries and is infeasible.
The goal is to identify a set of candidates from the domain $\Domain$,
such that it is computationally feasible to query the frequency oracle.
\subsection{Strawman Method}
\label{subsec:straw}


To better understand the protocols proposed in~\cite{Bassily2015local,rappor2}, we start by describing a strawman method for identifying a smaller set of candidates for frequent values.  An intuitive method is to divide and conquer.  Specifically, a length-$m$ value is divided into $g$ equal-size segments, each of length $s=m/g$.  For example, when $g=8,m=128$, each segment has $s=16$ bits.  Borrowing Python list syntax, we use the notation $v[i:j]$ to denote the segment of $v$ starting at the $i$'th bit and stopping at (and including) the $j-1$'th bit.  Thus $v[0:m]$ represents the complete $v$.

In the strawman protocol, each user randomly chooses a segment to report.  More specifically, the user randomly chooses $1\leq \alpha \leq g$, and reports
$$\langle\perturb(v), \alpha,\perturb(v[(\alpha-1)s:\alpha s])\rangle,$$
where $\perturb$ can be any perturb function, although it is natural to use \olh.  
That is, the whole population is divided (by their own random choices) into $g$ groups, each reporting on one segment.

The aggregator first queries the frequency of each length-$s$ binary string in each of the $g$ segments, issuing a total of $2^{s}\times g$ oracle queries, and identifying the frequent patterns in each segment.
Let $C_1,\ldots,C_g$ denote the frequent patterns for the $g$ segments.
The candidate set $C$ is the Cartesian product of $C_i$, i.e., $C=C_1\times C_2\times\ldots\times C_g$,
where Cartesian product operation $\times$ is defined as $C_i\times C_j=\{c_i||c_j:c_i\in C_i, c_j\in C_j\}$, and $||$ is the string concatenation operation. Finally, the aggregator queries frequencies of these candidates, using the full string reports $\perturb(v)$.

The main shortcoming of this method is that, if we identify $k$ candidates from each of $C_1,\ldots,C_g$, the candidate set $C$ has size $k^g$.  When $m$ is large, $g$ is not very small, and the candidate set $C$ is still too large to be enumerated.  The protocols proposed in~\cite{Bassily2015local,rappor2} can be viewed as taking two different approaches in further improving this method.


\subsection{The Segment Pairs Method (\spm)~\cite{rappor2}}

The approach taken by Google's team for the RAPPOR system improves upon the above strawman protocol by having each user report a pair of two randomly chosen segments, instead of reporting only one segment.  We call this the Segment Pair Method (\spm).

In \spm, the length-$m$ value is divided into $g$ segments of length $s=m/g$.
%
In addition to reporting the overall value $v$, a user also randomly chooses two segments to report.  More specifically, the user randomly chooses $1\leq \alpha \ne \beta \leq g$, and reports
 $$\langle\perturb(v),\alpha,\beta,\perturb(v[(\alpha-1)s:\alpha s]),\perturb(v[(\beta-1)s:\beta s])\rangle.$$
That is, the user runs three reporting protocols in parallel, each using one third of privacy budget.  Since each user randomly chooses $2$ out of $g$ segments to report, the population is divided into ${g\choose 2}$ groups, each reporting for one pair of segments.  
When $n$ users are reporting, one expects that about $\frac{n}{g/2}$ users report on each segment, and about $\frac{n}{g(g-1)/2}$ users report each pair of segments.

%
%

The aggregator first identifies the frequent patterns in each of the $g$ segments.  Then, it queries, for each pair $1\leq i,j \leq g$ of segments, the frequency for the values in $C_i\times C_j$ and identifies the value pairs that are frequent in segments $i,j$.
From the frequent value pairs for each pair of segments, the aggregator recovers candidates for frequent values for the whole domain, using the a priori principle that if a value $v\in \Domain$ is frequent, every pair of its segments must also be frequent.
Because of this filtering by segment pairs, the size of $C$ is typically small enough to query the frequency of each value in it.

The main limitation of this method is that, since the length of each segment must be relatively small (one needs to enumerate through all possible values for each segment), when the domain is large, there are too many pairs of segments.  As a result, the number of users reporting on each location-pair is limited, making it difficult to accurately identify frequent value pairs.  For example, when $g=8$, each pair has only about $\frac{n}{28}$ users.  And when $g=16$, each pair has only about $\frac{n}{120}$ users.

\begin{figure*}[t]
	\centering
	\subfigure[Strawman
				: Users partitioned into ${4 \choose 1}=4$ groups, each reporting one segment
				]{
		\label{subfig:strawman}
		\includegraphics[width=0.3\columnwidth]{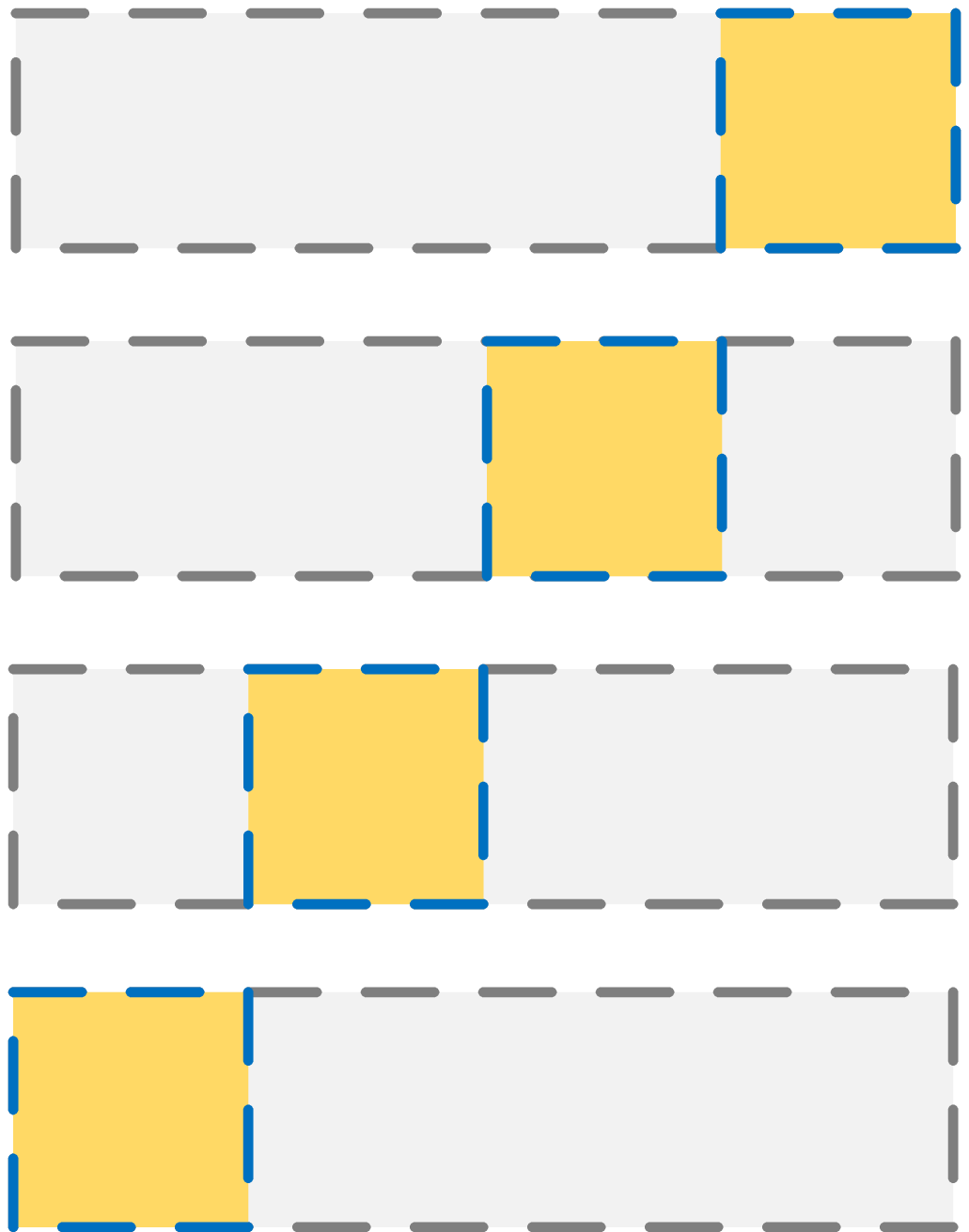}
	}
	\hspace{1.5cm}
	\subfigure[\spm
				: Users partitioned into ${4 \choose 2}=6$ groups, each reporting one pair of segments.
				]{
		\label{subfig:spm}
		\includegraphics[width=0.3\columnwidth]{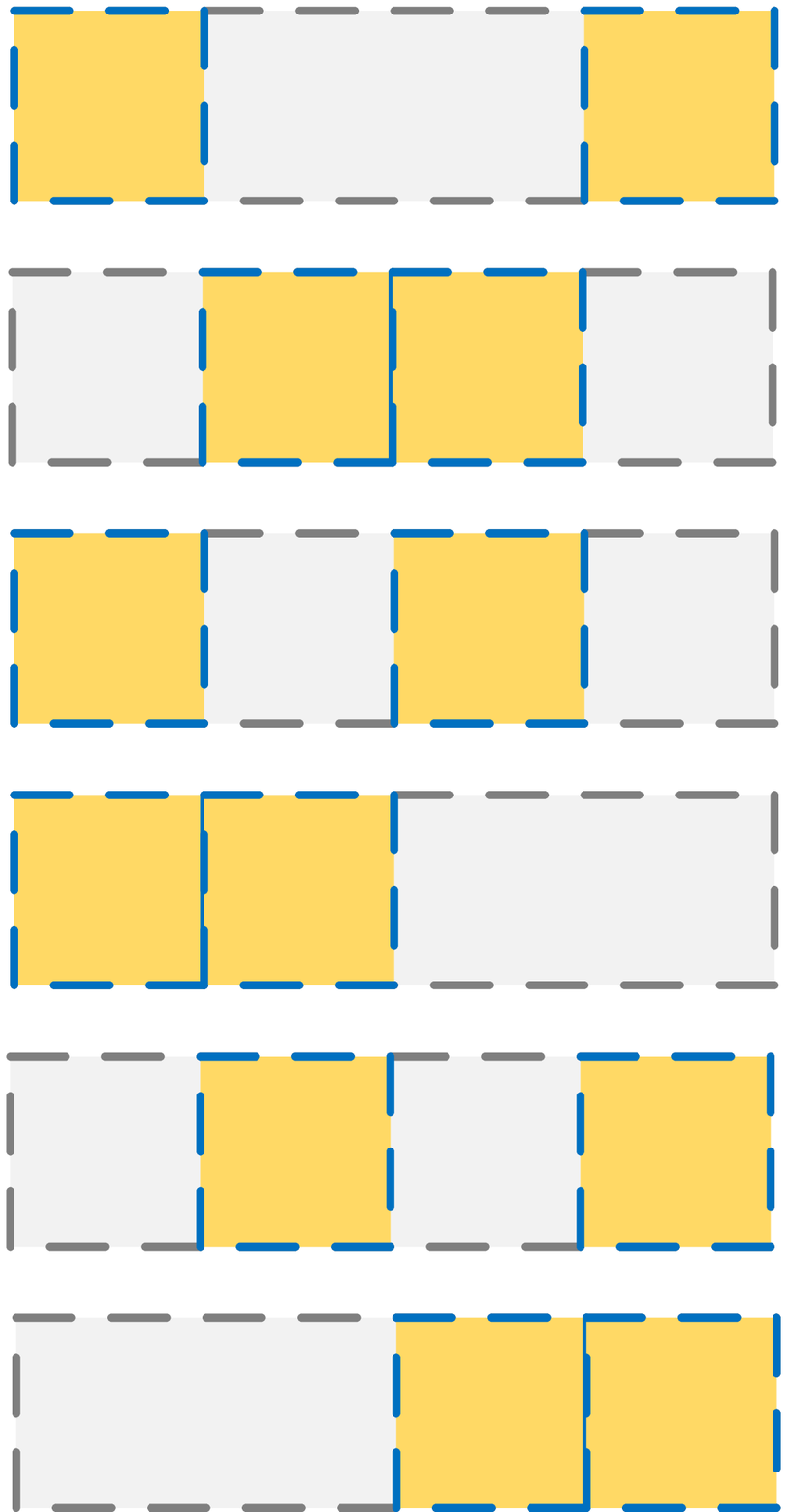}
	}
	\hspace{1.5cm}
	\subfigure[\hash
				: Users partitioned into $64$ groups, each reporting one bit on multiple channels.
				]{
		\label{subfig:hash}
		\includegraphics[width=0.375\columnwidth]{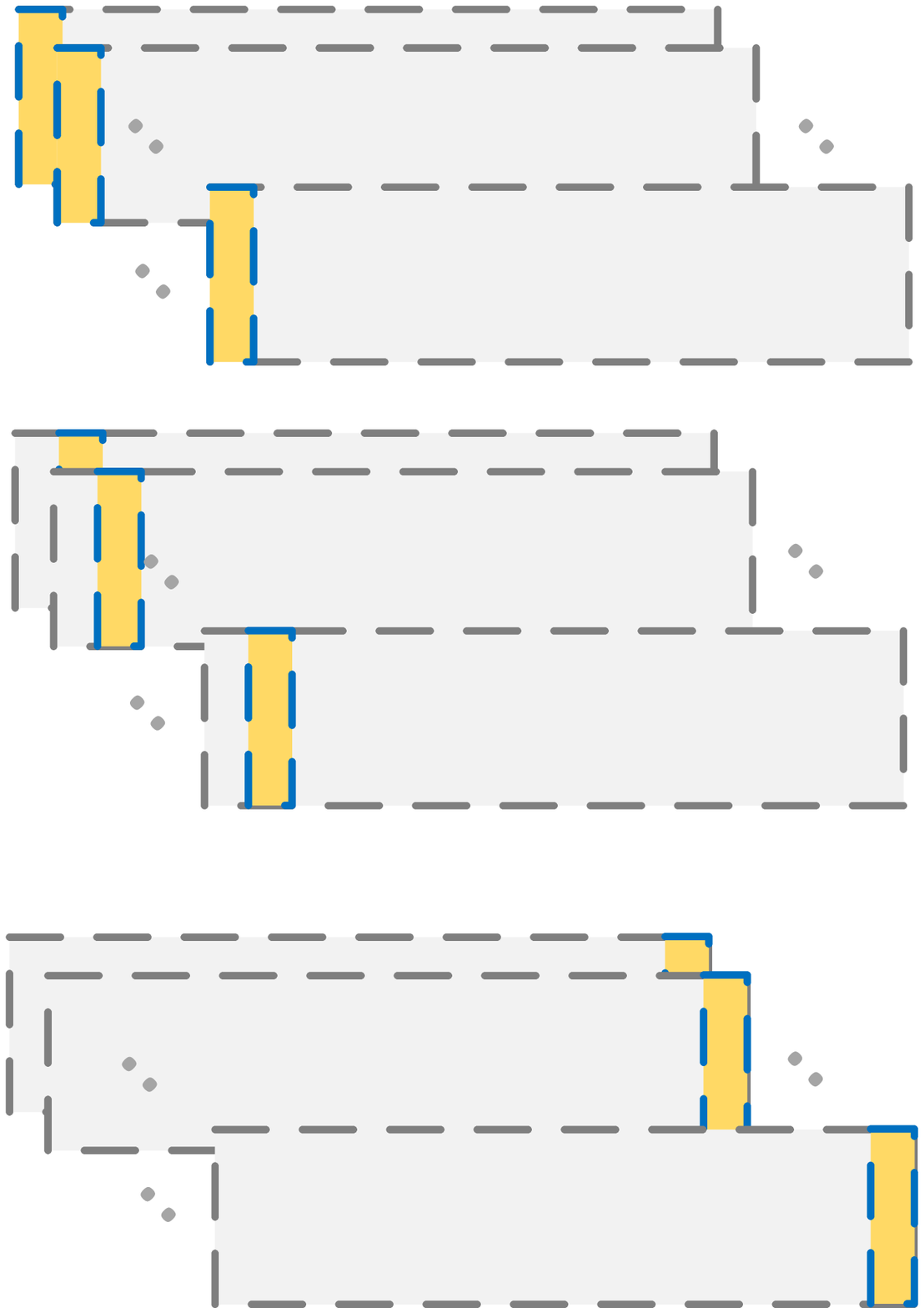}
	}
	\hspace{1.5cm}
	\subfigure[\itergroup
				: Users partitioned into $4$ groups, each reporting a prefix.
				]{
		\label{subfig:iterprod}
		\includegraphics[width=0.3\columnwidth]{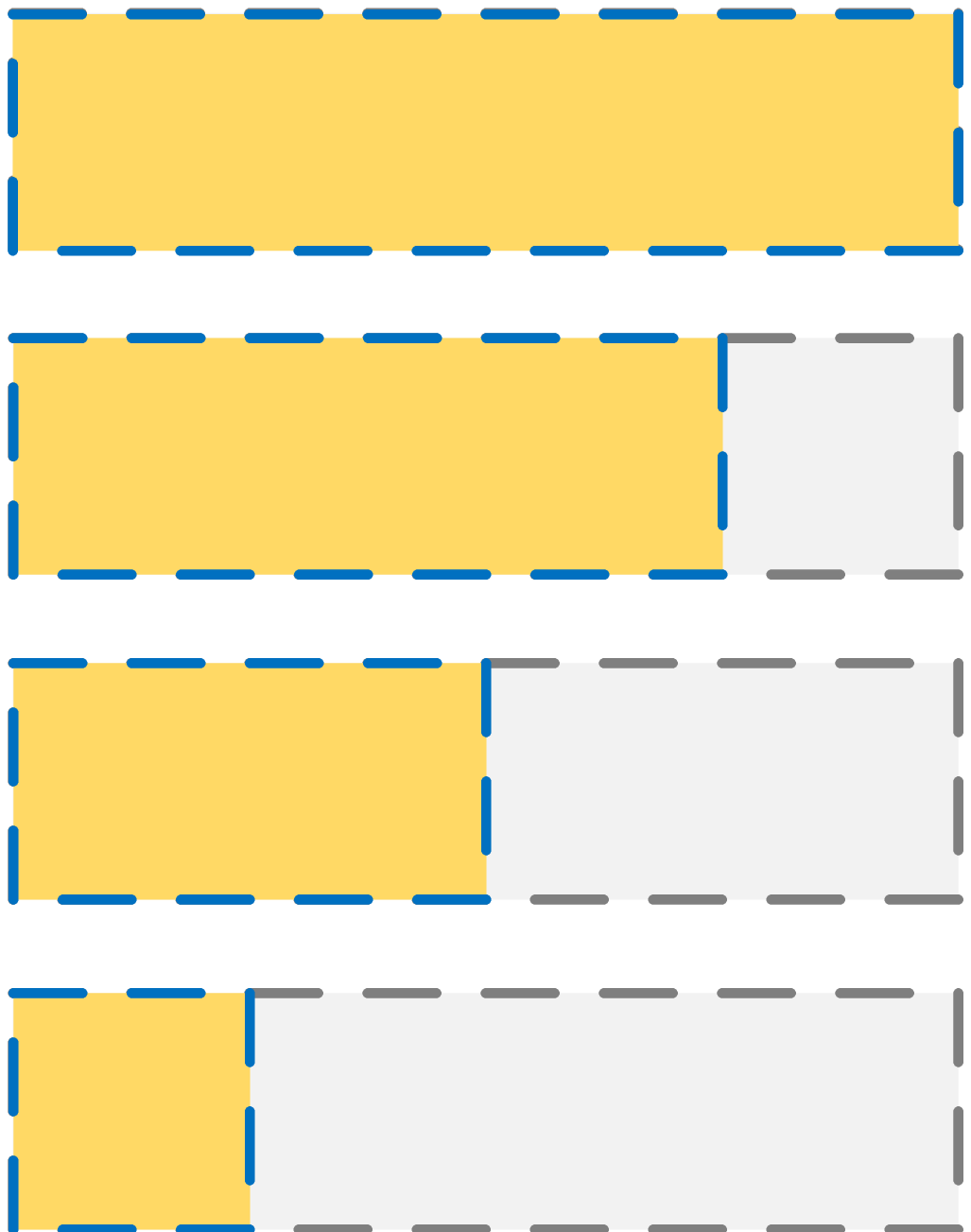}
	}
	\caption{
		Illustration of how candidate frequent values are generated by different methods.  Assuming that each value has $64$ bits, often divided into $4$ segments of $16$-bits each.  The first three methods require in addition the ability of to estimate the frequency of any value in the whole domain.  In \spm and \hash, this is done by having each user report both the segments in the figure and the whole value, dividing the privacy budget.  We show that it is better to divide the population to have another group of users who report only the whole value.  In \itergroup, the last group serves the purpose.
	}
	\label{fig:all_methods}
\end{figure*}

\subsection{The Multiple Channel Method (\hash)~\cite{Bassily2015local}}
\label{subsec:hash}

Bassily and Smith proposed an approach which we call Multiple Channel Method (\hash)~\cite{Bassily2015local}.  Our description of \hash below simplifies that in~\cite{Bassily2015local}, and is equivalent to it.  This approach can be viewed as improving upon the strawman approach by using a technique to separate the values into multiple channels so that with high probability each channel has at most one frequent value, and then identifying this candidate frequent value by identifying each bit of it.

The approach uses a hash function $H$ that maps each input value $v$ to an integer in $\{1\ldots h\}$.  We say that $v$ is mapped to the channel $H(v)$.  The value $h$ needs to be large enough to ensure that the probability that any two frequent values are mapped to the same channel is low.  Each user with input $v$ randomly selects $\ell$ such that $0\leq \ell <m$  and reports:
$$\tuple{\perturb(v), \ell, b_1,b_2,\cdots,b_h}$$
The privacy budget $\epsilon$ is divided into two parts $\epsilon_1+\epsilon_2=\epsilon$.  Sending the value $v$ uses $\epsilon_1$; $b_1,b_2,\cdots b_h$ are computed such that when $j\ne H(v)$, $b$ is a randomly sampled bit, and when $j=H(v)$, $b$ is a perturbed value of the $v[\ell]$, flipped with probability $q= \frac{1}{e^{\epsilon_2}+1}$.
That is, each user chooses one of the $m$ bit to report.  

From each channel, the aggregator extracts a candidate frequent value by taking the majority vote for each bit.  The aggregator then queries the frequency of these candidates and outputs the frequent values.

One main limitation of this approach is that since each user reports a single bit, only a small number of users are reporting for each bit.  For example, with $m=128$, only $\frac{n}{128}$ users participate in the determination of candidate for each bit.  Furthermore, to correctly recover the candidate value, each of the $128$ bits must be recovered correctly.  (While error correction code is suggested in~\cite{Bassily2015local}, that will further reduce the group size and increase the probability that any one bit is recovered correctly.)  This limitation can be addressed by having each user report a bigger block (such as 16 bit) at a time, which does improve the accuracy.

Another limitation is that since one identifies a single candidate from each channel, each user has to report on multiple channels, and the oracle queries must be made on all $h$ channels.  This adds a multiplicative factor of $h$ to the communication and computation overheads.

\section{Proposed Solution}
\label{sec:methods}




In both \spm and \hash, to deal with the challenge of large domains, a bit string input is divided into \emph{non-overlapping} segments so that one can recover frequent patterns in each segment.  These patterns need to be combined into a set of candidate frequent values.  \spm does this by making each user report a pair of segments, dividing the population into ${g \choose 2}$ groups.  \hash does this by using multiple channels so that within each channel one focuses on identifying a single candidate frequent value.

We observe that instead of dividing a bit string into non-overlapping segments, one can have these segments \emph{overlapping}.  
In our proposed method, which we call Prefix Extending Method (\itergroup), users in each group report a prefix of her value.  
Figure~\ref{fig:all_methods} illustrates the differences between the four methods we have discussed.  The main advantage of \itergroup over other methods is that when $m$ is long, one needs to divide the population only into $g$ groups.  

\subsection{Prefix Extending Method (\itergroup)}
\label{subsec:group}

The \itergroup method is parameterized by two parameters $\gamma$ and $\eta$, which are positive integers.
A user is randomly assigned into one of $g$ groups, where $g=\lceil\frac{m-\gamma}{\eta}\rceil$.  The assignment can be made by the aggregator, or having each user selecting a group at random.  Users in the $i$'th group where $1\leq i\leq g$ report
$$\langle i,\perturb(v[0:\gamma+i\eta])\rangle.$$

Let $D_1=\{0,1\}^{\gamma+\eta}$, the aggregator uses the first group's reports to identify which values in $D_1$ are frequent prefixes.  Let $C_1$ be the result.  It then constructs $D_2=C_1 \times \{0,1\}^\eta$, which are candidates for longer frequent prefixes, and uses the second group's reports to identify the frequent ones in $D_2$ as $C_2$.  This continues until the last step where $C_g$ gives the set of frequent values.

Note that here we assume that we have no domain knowledge about the underlying values and thus represent the values as bit strings and divide it into equal-length segments.  The basic idea of \itergroup, where one iteratively find portions of the whole values that are frequent, can be applied in other contexts.  In a given application, one can take advantage of  domain-specific knowledge to define segments differently.  For example, one can use the domain knowledge to eliminate candidates that are impossible.  If the values have internal structures such as one component can have values of different lengths, one can also extend that component in one step and test values of different lengths.  In this paper, we focus on the binary string setting.

\subsection{Protocol Analysis}
\label{sec:analyze}

The \itergroup protocol has two parameters $\gamma$ and $\eta$.  Typically, $\gamma$ should be slightly larger than $\eta$, to make the 
candidate set size roughly the same in each step.  The choice of $\eta$, however, is very important.  Larger $\eta$ would mean higher computational cost.  
Furthermore, by having a large $\eta$, there will be fewer groups, and thus more users in each group, making the estimation in each step more accurate;
on the other hand, there will be more values to consider in each step, thus the probability a non-heavy hitter is identified is increased.
    %
We now analyze the utility to optimize the choice of $\eta$.


\subsubsection{Metric}
To compare utility when using different $\eta$ values, we use the following utility measurements.

\textbf{F-measure (\fone)}.
Define $v_j$ as the $j$-th most frequent value.
The ground truth for top $k$ values is $C_T=\{v_1,v_2,\ldots,v_k\}$.
Denote the $k$ values identified by the protocol using $C_g$.
$C_T\cap C_g$ is the set of real top-$k$ values that are identified by the protocol,
and $C_T\cup C_g$ is the union of the two sets.
We use the widely used F-measure~\cite{manning2008introduction} which
is the harmonic mean of precision and recall, i.e.,
\begin{align*}
&\fone=\frac{2}{{1}/{P}+{1}/{R}}=\frac{2PR}{P+R}\\
\mbox{where }&P=\frac{|C_T\cap C_g|}{|C_g|}, R=\frac{|C_T\cap C_g|}{|C_T|}
\end{align*}

We note that when $|C_T| = |C_g|$, the precision $P$ equals the recall $R$,
and the F-measure equals the precision, as well as 1 minus the false negative rate.

\textbf{Normalized Cumulative Rank ($\mathsf{NCR}$)}.
The F-measure uses only the unordered set $C_T$ as the ground truth.
As a result,
missing the value with the highest frequency is penalized the same as missing any others.
To address this limitation, we assign a quality function $q(\cdot)$ to each value,
and use the Normalized Cumulative Gain ($\mathsf{NCG}$) metric~\cite{jarvelin2002cumulated}:

\begin{align*}
{\mathsf{NCG}} = \frac{\sum_{v\in C_g} q(v)}{\sum_{v\in C_T} q(v)}.
\end{align*}

We instantiate the quality function using $v$'s rank as follows:
the highest ranked value has a score of $k$ (i.e., $q(v_1)=k$), the next one has score $k-1$, and so on;
the $k$-th value has a score of 1, and all other values have scores of 0.
To normalize this into a value between 0 and 1,
we divide the sum of scores by
the maximum possible score, i.e., $\frac{k(k+1)}{2}$.
This gives rise to what we call the Normalized Cumulative Rank ($\mathsf{NCR}$);
this metric uses the true rank information of the top-$k$ values.

Both F-measure and $\mathsf{NCR}$ are in the
range $[0.0, 1.0]$, where higher values indicate better accuracy.
We present results using these metrics and observe that the correlation
among them is quite stable.

\textbf{Unified Utility Score.}
\label{subsec:opt}
We express the utility scores as the weighted average of the identification probability of each heavy hitter, that is,
\begin{align}
\sum_{j=1}^{k}\left(w_j\cdot\prod_{i=1}^{g}\Pall{i}{j}\right)\label{eqn:utility_score},
\end{align}
where $\Pall{i}{j}$ is the identification probability for the $j$-th most frequent value $v_j$ in step $i$, that is, $v_j[0:\gamma+i\eta]\in C_i$.
We will elaborate $\Pall{i}{j}$ later.
The overall identification probability is the product of that of each phase.

Different metrics can be expressed by different weights.
In the F-measure, $w_j=\frac{1}{k}$, and for $\mathsf{NCR}$,
where the higher ranked value receives greater weight,
$w_j=\frac{k+1-j}{\sum_{l=1}^{k}k+1-l}$.

In our analysis below, we assume that the identification of the heavy hitters are mutually-independent.  Technically this is not true.  
If one value has been identified as a heavy bitter, the probability that another one is identified will be slightly lower, since we are identifying $k$ heavy hitters. 
However, when $k$ is not very small, this effect is small and can be ignored for our purpose.  We will empirically verify the correctness of this approximation.

\subsubsection{Assumptions and Constraints}
\label{subsec:constraints}
We first simplify \itergroup by initiate some parameters.
Recall that there are two main parameters $\gamma, \eta$ in \itergroup.
Users are partitioned into $g=\lceil\frac{m-\gamma}{\eta}\rceil$ groups.
We assume that the number of users in all groups are the same; thus, the number of users in the $i$'th group is $n[i]=n/g$. 
We fix the size of output in each stage to be $|C_i|=k$.  We further fix $\gamma=\lceil \log_2 k \rceil$.
Thus the aggregator makes $|D_i|=2^{\gamma+\eta}=k\cdot 2^\eta$ queries to the frequency oracle in each step,
and the only parameter left for us to choose is $\eta$.

To calculate the actual utility scores, we have to make some assumptions of the dataset distribution.
This is because the significance (frequency) of the heavy hitters will affect the utility measure.
For example, in an almost uniformly distributed dataset, it is hard to find out the most frequent $k$ values,
since the frequency differences are very small.
We also limit the maximum allowed number of frequency oracle queries.  

\subsubsection{Approximate Identification Probability}
\label{subsec:iden_prob}

We now calculate $\Pall{i}{j}$, the probability $v_j$'s prefix is identified in step $i$, i.e., $v_j[0:\gamma+i\eta]\in C_i$.

We first show the estimation of a value is a random variable.
Assume the true frequency of $v_j$ is $f_j$.
$n[i]=n/g$ users are randomly assigned to report the first $\gamma+i\eta$ bits of their private value,
and each of them possesses the value $v_j$ with probability $f_j$.
By \eqref{eqn:est}, 
estimation of $v_j$ is only determined by the ``support'', $I_{j}$, it receives.
Since we only care about the relative ranks of the estimations,
we focus on $I_j$ and use it as estimation of $v_j$.
For each user, if he has value $v_j$ (with probability $f_j$),
his report will ``support'' $v_j$ (reports $H^j(v_j)$ in terms of \olh) with probability $p$;
otherwise his report will ``support'' $v_j$ with probability $q=1/d'$. 
Therefore, $I_{j}$ can be seen as the summation of $n[i]$ binomial variables,
whose probability of being $1$ is $p_j =
p\cdot f_j+q\cdot(1-f_j)$.
In most cases (as long as $n[i]$ is large),
we can approximate $I_j$ using normal distribution with
mean $\mu_j[i]=n[i]\cdot p_j$
and variance $\sigma^2_j[i]=n[i]\cdot p_j\cdot (1-p_j)$.

We then calculate the probability $v_j[0:\gamma+i\eta]$'s estimation is ranked on top $k$.
Since $I_j$ is a normal random variable,
we know the probability a value is estimated above any threshold value $T$,
that is, $\Pr{I_j>T}$.
For all the non-heavy hitters,
summing this up gives us the expected number of values that are estimated above $T$.
If this expected number is less than $k$,
then $\Pr{I_j>T}$ is the probability its estimation is ranked on top $k$.
We use $T_k[i]$ to denote this threshold value.
For efficiency, we assume that among all the values to be tested, $D_i$,
all the $N[i]=|D_i|-k=k(2^\eta-1)$ non-heavy values have zero frequencies
(this is especially safe when $N[i]$ is large).
Therefore, $T_k[i]$ can be calculated by the inverse of cumulative density function:
\begin{align*}
T_k[i]=&-\Phi^{-1}\left(\frac{k}{N[i]}\right)\cdot \sigma_0[i]+\mu_0[i]
\end{align*}
where $\sigma_0^2[i]=n[i]\cdot q\cdot (1-q)$, $\mu_0[i]=n[i]\cdot q$ denotes the variance and mean of these $N[i]$ zero-mean values.

Finally, we account for the effect of other heavy hitters.
We assume the estimations of the top $k$ values are always sorted,
thus for the $j$-th value to be estimated top $k$,
there are $k-j+1$ slots,
since the top $j-1$ values are already ranked higher.
Formally,
\begin{align}
\Pall{i}{j}=&1-\Phi\left(\frac{T_{k-j}[i]-\mu_j[i]}{\sigma_j[i]}\right)=\Phi\left(\frac{\mu_j[i]-T_{k-j}[i]}{\sigma_j[i]}\right)\nonumber\\
=&\Phi\left(\frac{\mu_j[i]+\Phi^{-1}\left(\frac{k-j}{N[i]}\right)\cdot \sigma_0[i]-\mu_0[i]}{\sigma_j[i]}\right)\label{eqn:iden_prob}
\end{align}

\subsection{Instantiate \itergroup}
%

%

With \eqref{eqn:iden_prob} to instantiate \eqref{eqn:utility_score},
we can now calculate the utility scores when using different values for $\eta$.
Before the actual numerical computation, we need to assume the dataset distribution.
For instance, we can assume that the users' value form a zipf's distribution,
i.e., $f_j\propto \frac{1}{j}$.
We also need to limit the number of total queries to the frequency estimator,
e.g., $2^{20}$ queries.

The optimization inputs are:
$k$ as the number of desired heavy hitters,
$m$ as the domain size in bits,
$n$ as the number of users,
and $\epsilon$ as the privacy budget.
With all parameters available,
$\eta$ is instantiated with different values and the corresponding utility scores are calculated.
The configuration that gives best utility score will be used.

In the actual optimization, we can make small changes to more parameters,
namely, $\gamma$, and in different steps $i$, the number of users $n[i]$,
and the candidate size $|C_i|$. It is also possible to use different $\eta$ for each step,
denoted by $\eta[i]$.
It turns out that slightly change in these parameters does not affect the final result much.
What affect utility the most is the number of groups one needs to divide the users into.
When $\eta$ is very small, the overall utility will deteriorate a lot.  

As a result, in most of the system settings, the optimal configuration is 
$\gamma=\log_2 k$, and for all $i$, $|C_i|=k$, $\eta[i]=\eta$, $n[i]\approx n/g$,
where $\eta$ is the maximal integer such that the total number of queries
$2^{\gamma+\eta}\cdot g$ ($g=\lceil\frac{m-\gamma}{\eta}\rceil$)
is less than the limit.  In some extreme cases, e.g., when $k$ is very big,
$|C_i|$ is smaller than $k$, suggesting that when $k$ is too large, since it is impossible to accurately recover $k$ heavy hitters, 
one should simply try to find fewer.  On the other hand, when $k$ is small, $|C_i|$ can be large,
so that the probability prefixes of the heavy hitters are included in $C_i$ are increased.




\subsection{Observations and Design Principles}
\label{subsec:principle}

By the optimization and the supporting analysis,
we are able to answer the questions raised in the beginning of this section.
Specifically, we make some statements that can help guide the design of protocols for not only the heavy hitter problem but also other LDP problems.
%
%
%

\newcommand{\eex}{e^{\epsilon/Q}}
\newcommand{\ee}{e^{\epsilon}}
\newcommand{\eey}{e^{\epsilon-\epsilon/Q}}



\newcommand{\eer}{\frac{e^{\epsilon/r}-1}{\sqrt{e^{\epsilon/r}}}}
\newcommand{\eers}{\left(\frac{e^{\epsilon/r}-1}{\sqrt{e^{\epsilon/r}}}\right)^2}
\newcommand{\eenr}{\frac{e^{\epsilon}-1}{\sqrt{e^{\epsilon}}}}
\newcommand{\eenrs}{\left(\frac{e^{\epsilon}-1}{\sqrt{e^{\epsilon}}}\right)^2}

First of all, \itergroup assigns users to different groups,
and let each user report a prefix.
It is possible to design \itergroup such that requires each user answer all prefixes.
Since all prefixes from the same value are obviously related,
correlation information can be extracted.
By~\eqref{eqn:iden_prob}, we observe the following:
\begin{proposition}
	\label{thm:par_user_iden}
	By partitioning the users into groups and letting each group answer a separate question,
	the identification probability will be higher than having each user split privacy budget to answer all the questions.
\end{proposition}
We provide a proof to support our choice in the appendix.

We note that a similar observation has been made in~\cite{ldpprimitive}, based on comparing variances of frequency estimation.  
We can generalize these observations into the following principle for designing effective LDP protocols.  
\begin{principle}
In LDP setting, one should divide the user population, instead of dividing the privacy budget.
\end{principle}
We note that in the centralized DP setting, these are generally equivalent in terms of utility.  

From numerical computation of optimal parameters, we make the following observations. 
%





\begin{obersvation}
	\label{con:morebits}
	Within the query limit,
	greatest identification probability is achieved with the least number of groups.
\end{obersvation}

This implies that, when multiple groups are necessary, $\eta$ should be as big as possible
(within query limit).  
We distill this as a second principle for designing LDP protocols. 

\begin{principle}
When designing LDP protocols, one should minimize the number of groups one has to divide the user population into, as a large group size is a key in obtaining accurate answers.  
\end{principle}

A similar principle applies in the centralized setting as well.  By reducing the number of steps that need privacy budget, one can often achieve better results. 
In \itergroup, in order to increase the final utility,
it is inclined to allocate more users to the final group,
or minimize work by reduce $\eta$ in the final phase.
Somewhat counter-intuitively, a more balanced allocation will be better.

\begin{obersvation}
	\label{con:eq_user}
	If $\eta[i]=m/g$,
	best result is achieved when $n[i]=n/g$.
	Similarly, if $n[i]=n/g$,
	best result is achieved when $\eta[i]=m/g$.	
\end{obersvation}

The previous observation is made when one of $n[i]$ and $\eta[i]$ is balanced
(i.e., $n[i]=n/g$ or $\eta[i]=m/g$).
When either is unbalanced (perhaps due to e.g., the leftover bits in the final phase),
theoretical optimal allocation of the other is unknown.
These observations cannot be proven easily.
We provide empirical support in the next section
(specifically, Figure~\ref{fig:match_config}).

\section{Validation of Analysis}
\label{sec:ver}

In this section,
we run empirical experiment to validate the analysis.
That is, we want to show the analysis matches empirical results
(mainly~\eqref{eqn:utility_score} and~\eqref{eqn:iden_prob}).
Moreover, we verify the design principles we made as observations in Section~\ref{subsec:principle}.

%
%
We generate synthetic datasets follows zipf's distribution
(i.e., the $j$-th most frequent value has frequency $f_j$ proportional to $\frac{1}{j^{1.5}}$).
The values of the heavy hitters are random.
We then evaluate \itergroup on the synthetic datasets with reasonable parameters.
All experiment runs 100 times.
For each setting, we use points with error bar (if any) for empirical results, 
and a thin line with the same color for analytical results.
The validation is carried out in the following three stages.

\begin{figure*}[t]

	\subfigure[Different $n$]{
		\label{subfig:vary_n_multi}
		\includegraphics[width=0.66\columnwidth]{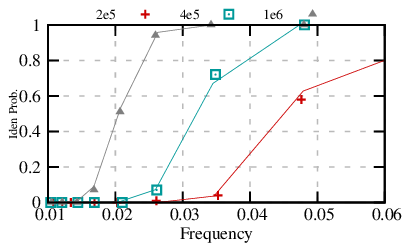}
	}
	\subfigure[Different $m$]{
		\label{subfig:vary_m_multi}
		\includegraphics[width=0.66\columnwidth]{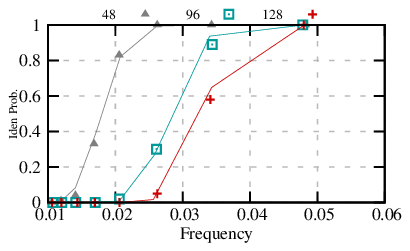}
	}
	\subfigure[Different $\epsilon$]{
		\label{subfig:vary_e_multi}
		\includegraphics[width=0.66\columnwidth]{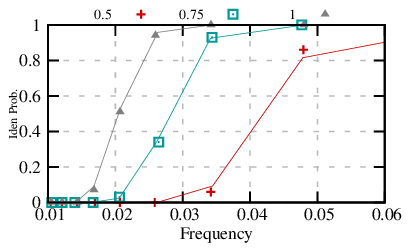}
	}
	\caption{
		Overall identification probabilities for values with different frequencies. 
        Dots represent empirical results, and lines shows analytical results.
		We fix $n=1000000,\epsilon=1,m=64$ as the default setting.
		In each sub-figure, we vary one of them while keeping the others. 
	}
	\label{fig:vary_f_multi}
\end{figure*}

\begin{figure*}[t]
	\subfigure[Vary $g$]{
        \label{subfig:vary_r}
        \includegraphics[width=0.66\columnwidth]{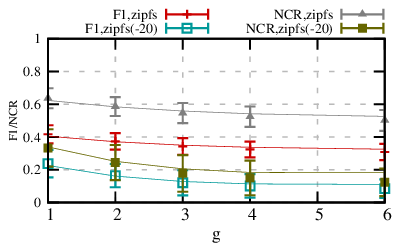}
    }
    \subfigure[Vary {$\eta[1]$}]{
        \label{subfig:vary_s}
        \includegraphics[width=0.66\columnwidth]{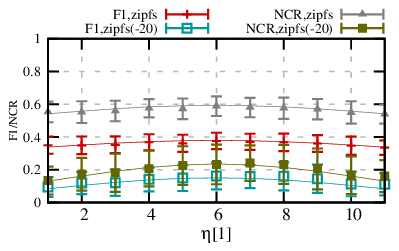}
    }
    \subfigure[Vary {$n[1]$}]{
        \label{subfig:vary_n1}
        \includegraphics[width=0.66\columnwidth]{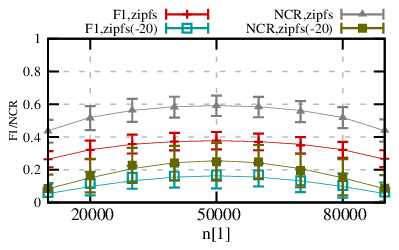}
    }
    
   	\caption{
   		Utility score from different configurations.
        We fix $k=16$, and use the default configuration of $n=100000, m=16,\epsilon=1$. 
        In each sub-figure, we vary $r, s[1], n[1]$,
        and plot \fone, \ncr for two distributions.
    }
    \label{fig:match_config}
\end{figure*}

\begin{figure*}[t]
	\subfigure[Vary $n$]{
		\label{subfig:vary_n}
		\includegraphics[width=0.66\columnwidth]{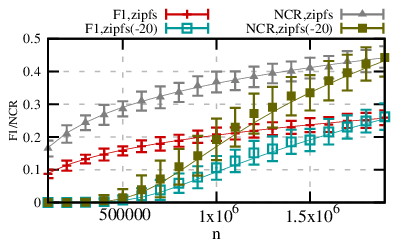}
	}
	\subfigure[Vary $m$]{
		\label{subfig:vary_m}
		\includegraphics[width=0.66\columnwidth]{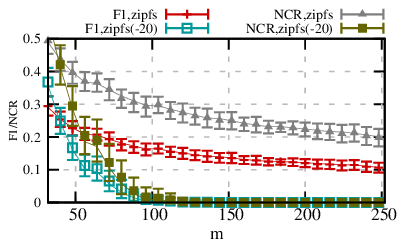}
	}
	\subfigure[Vary $\epsilon$]{
		\label{subfig:vary_e}
		\includegraphics[width=0.66\columnwidth]{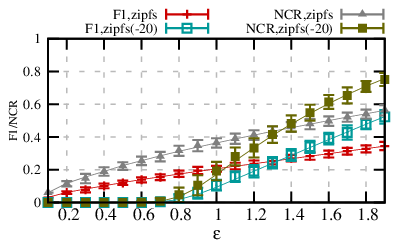}
	}
	
	\caption{
		Utility score from the optimal configuration.
		We fix $k=32$, and use the default configuration of $n=1000000, m=64,\epsilon=1$. 
		In each sub-figure, we vary $n,m,\epsilon$,
		and plot \fone, \ncr for two distributions.
	}
	\label{fig:match_optimal_config}
\end{figure*}

\subsection{Identification Probability.}
%
We first verify the correctness of the identification probability (i.e., \eqref{eqn:iden_prob}).
We generate $n=100000$ data points each with $m=64$ bits.
Users report with privacy budget $\epsilon=1$.


Note that with $m=64$ bits or more, it is not feasible to estimate the whole domain in a single round.
For the purpose of demonstration, 
we instantiate \itergroup with $\gamma=5,\eta=10,|C_i|=32$.

In Figure~\ref{fig:vary_f_multi}, 
we show the identification probability of values with different frequencies.
We observe that, first of all, the analysis matches the empirical result pretty well.
Moreover, when the other parameters remain the same,
the more users (larger $n$),
the shorter the domain length (smaller $m$),
or the larger the privacy budget (larger $\epsilon$), 
the better the overall result.
This trend also matches the analysis.

\subsection{Utility Scores under Different Configurations.}
Having verified the analytical identification probability~\eqref{eqn:iden_prob} for each single value,
we now verify the analytical utility score~\eqref{eqn:utility_score}, which is the linear combination of the identification probabilities for each heavy hitter.


We generate $n=100000$ data points each represented by $m=16$ bits,
and use $\epsilon=1$.
Besides the zipfs distribution, 
we also test on another similar distribution, `zipfs(-20)'.
The difference is that in zipfs(-20), the most frequent 20 values are dropped.
As a result, the first few heavy hitters are not that significantly frequent,
but the following ones occurs more frequently.

The domain of this dataset is made smaller mainly for the purpose of making it clear when we compare different configurations of \itergroup.
We plot $\fone$ and $\ncr$ scores of the top $k=16$ heavy hitters in Figure~\ref{fig:match_config}.
It can be seen that all points (empirical results) lie on the line (analytical results),
verifying that the analysis is accurate under different configurations.

Specifically, in Figure~\ref{subfig:vary_r},
we vary the number of rounds from $1$ to $6$, each outputting $16$ candidates.
As a result, with more rounds of test, the overall result becomes worse.
This also validates principle  
(Observation~\ref{con:morebits}) to have as few rounds as possible.

To support Observation~\ref{con:eq_user},
we fix $|C_i|=k=16$, $\gamma=4$ and $g=2$,
we run two sets of experiment.
In the first one, we assign half of the users to each group,
and vary segment size of the first round.
It can be seen from Figure~\ref{subfig:vary_s} that when $\eta[1]=6$ 
(both rounds analyzes a domain of $2^{10}$ values), 
the overall result is optimal.
Moreover, the result is symmetric on $\eta[1]=6$
(since $\gamma=4$, the effective number of bits to examine is $12$).
This is also the same as what we expect.
In the second experiment,
we fix $\eta=6$ and try different user allocations , i.e., $n[1]$.
As shown in Figure~\ref{subfig:vary_n1},
when each of the two rounds receives half of the users, 
the overall result is maximized.
Similar to the result for segment size, the result is symmetric on $n[1]=50000$.


\subsection{Optimal Utility Scores under Different Scenarios.}

Now we verify the correctness of the optimal configuration under different scenarios.
We use the default setting of $n=1000000,m=64,\epsilon=1$, 
and plot $\fone$ and $\ncr$ values of optimal \itergroup, 
varying any of $n,m,$ and $\epsilon$.
The configuration is optimized taking \fone as the goal.

From Figure~\ref{fig:match_optimal_config}, 
we can first confirm analysis still matches empirical results well.
When the setting is not ``favorable''
($\epsilon$ is small, $n$ is small, or $m$ is large),
zipfs distribution has better results, 
while zipfs(-20) gives similar or better result on the other extreme.
The reason is that, 
when $m$ is large and $\epsilon, n$ are not sufficiently big,
the noise in \itergroup is large.
In zipfs distribution, the first several heavy hitters are more significant,
and therefore, the overall utility score is better.
While on the other hand,
it is possible to recover more heavy hitters.
The following several heavy hitters in zipfs(-20) distribution,
which are more frequent,
contribute more to the overall result.

\section{Evaluation}
\label{sec:eval}
Now we discuss experiments that evaluate different protocols.
Basically, we want to answer the following questions:
First, how many heavy hitters can be effectively identified.
Second, how much improvement is \itergroup over existing protocols.
Finally, what are the effects of different design choices in \itergroup.



\subsection{Evaluation Setup}  

Each experiment is run 10 times, and the average and standard deviation are reported.

\mypara{Datasets.}
The following three datasets are used.
We assume the zipf's distribution when optimizing \itergroup.
Note that in the real world, 
auxiliary information (heavy hitter dictionary) may exist to help improve the result.
For example, the system BLENDER~\cite{avent2017blender} is proposed to work under the assumption that a certain amount of users will participate in a centralized DP protocol to find out the dictionary of heavy hitters.
However, our focus is on the case where there is no additional dictionary or the heavy hitters are changing frequently so that existing dictionaries are not reliable to provide up-to-date information.


%

\textbf{1. Frequent URL.}
In \rappor~\cite{rappor2}, the authors synthesized
one million urls from a confidential distribution of only 100 websites.
The urls are fixed to be $20$ bytes (160 bits) long
(padding or truncating if needed).
We mimic the distribution by collecting a similar dataset from Quantcast~\cite{quantcast}.
The dataset contains domain name and monthly visited people of the $80$ thousand most frequently visited websites. 
We limit urls to $20$ bytes and limit the analysis to a $5$-minute period, 
resulting a dataset containing $1.2$ million data points, and $27$ thousand unique urls.

This also motivates a real-world application,
where the analyst can find out the most popular website.
In a previous report of RAPPOR~\cite{rappor},
the system collects homepage urls from users,
by testing on a known list of websites.
We focus on the scenario where the dictionary is unavailable or inaccurate.

\textbf{2. Query Trends.}
The AOL dataset contains user queries on AOL website during the first three months in 2006.
Similar to the settings of~\cite{avent2017blender}, 
we assume each user reports one query (w.l.o.g., the first query).
The queries are limited to be $6$ bytes long.
This results a dataset of around $0.5$ million queries including $0.2$ million unique ones.

Many real-world application such as keyword trends or hot tags can be derived from this example.
In these scenarios, the heavy hitters change frequently, 
such that the dictionary from history may not be reliable.
%

\textbf{3. Synthetic Dataset.}
We generate a synthetic dataset of $n=1000000$ data points following the exponential distribution
(also known as geometric distribution).
The values (heavy hitters) are randomly distributed.
Each value is represented by $m=64$ bits.
The exponential scale is $0.05$, which is close to the experimental setting in~\cite{rappor}.

\mypara{Competitors.}
We consider the following algorithms: \itergroup, \hash, and \spm.
In order to optimize \itergroup, we assume a zipf's distribution,
limit the number of queries to the frequency estimator to $2^{20}$,
and take \fone as the goal.

Both \spm and \hash were designed to find heavy hitters based on threshold,
but \itergroup works for top $k$ heavy hitters. 
For a fair comparison, we improve \hash and \spm in step 1 and 2,
and change them from threshold based algorithms to top $k$ based in step 3.
Note that \itergroup can also be changed to work for threshold. 
The corresponding results are shown in Section~\ref{sec:thres}.

\textbf{1. Replace LDP primitive.}
Existing methods use non-optimal LDP primitives,
but they can be changed.
Specifically, \spm use \rap~\cite{rappor} as the internal LDP primitive, 
and \hash uses \blh. 
We replace \rap and \blh with \olh to improve their efficiency.

\textbf{2. Reduce Number of Groups.}
For the url dataset, \spm specifies one segment length to be two bytes.
But for other domain length, there is no clear specification as to how long each segment should be.
Guided by Observation~\ref{con:morebits},
we make the segment length as long as possible,
under the frequency oracle query limit.

\hash uses $n^{1.5}$ channel, which is infeasible in many scenarios.
We observe that collision of other non-frequent values does not effect much, 
and propose to use $k^{1.5}$ channels.


\textbf{3. Replace Threshold Test.}
Existing methods require internal test and filtering based on a threshold.
Specifically, there is a final testing phase on all the identified values.
Only those tested above a threshold will be returned.
We replace this constraint by releasing the top $k$ values for a fair comparison.

In \rappor, moreover, 
each segment or segment-pair will be identified if its frequency is estimated above a threshold.
We relax this by limiting exactly $k$ patterns in each segment.
This ensures to identify at least $k$ heavy hitters.
For the location-pairs,
we 
keep adding segment-pairs until more than $k$ candidates are identified.

\subsection{Detailed Results}

\subsubsection{Effect of $\epsilon$}
We show \fone and \ncr results of different methods varying $\epsilon$ in Figure~\ref{fig:eval_vary_method}.
It is clear that \itergroup performs best among all three protocols.
When $\epsilon$ increases,
the number of heavy hitters that can be identified will increase.
The improvement is more significant when $\epsilon$ is larger.

When $\epsilon=4$, \itergroup achieves $\fone=0.9$,
meaning that more than ten frequent URLs can be identified;
on the other hand, \hash and \spm can only identify two.

\begin{figure*}[t]
	\subfigure[URL, \fone]{
		\label{subfig:eval_eps_url}
		\includegraphics[width=0.67\columnwidth]{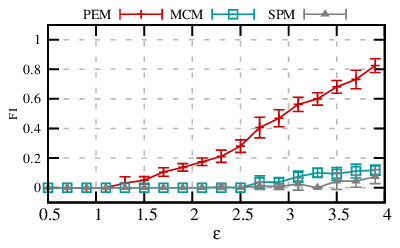}
	}
	\subfigure[AOL, \fone]{
		\label{subfig:eval_eps_password}
		\includegraphics[width=0.67\columnwidth]{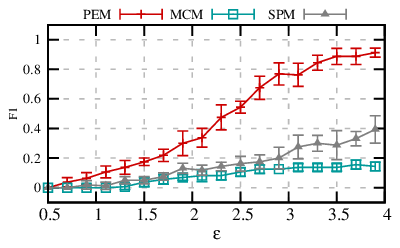}
	}
	\subfigure[Synthetic, \fone]{
        \label{subfig:eval_eps_synthesizeexponential}
        \includegraphics[width=0.67\columnwidth]{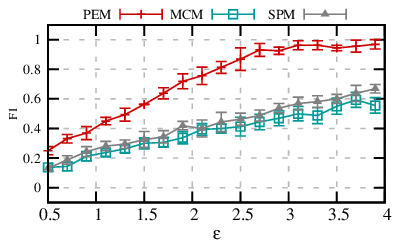}
    }
	
	\subfigure[URL, \ncr]{
		\label{subfig:eval_eps_url_ncr}
		\includegraphics[width=0.67\columnwidth]{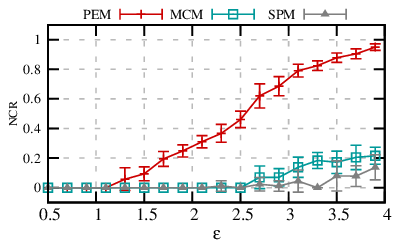}
	}
	\subfigure[AOL, \ncr]{
		\label{subfig:eval_eps_query_ncr}
		\includegraphics[width=0.67\columnwidth]{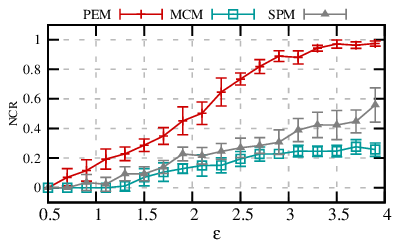}
	}
	\subfigure[Synthetic, \ncr]{
		\label{subfig:eval_eps_synthesizeexponential_ncr}
		\includegraphics[width=0.67\columnwidth]{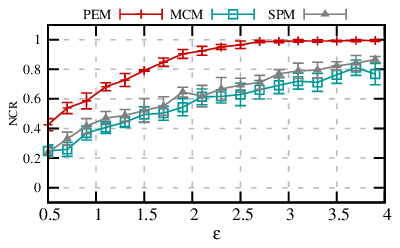}
	}
	\caption{Evaluation of the datasets, vary $\epsilon$ while fixing $k=16$. }
	\label{fig:eval_vary_method}
\end{figure*}

\subsubsection{Effect of $k$}
Figure~\ref{fig:eval_vary_k_method_1-32} gives \fone and \ncr results of different methods varying $k$.
Similarly, we can see that \itergroup outperforms \hash and \spm.
Note the correlation between \fone and \ncr is close:
a protocol with better \fone score will also have a better \ncr score.
Thus, from now on, we ignore the \ncr scores.

For most of the cases, utility scores decrease with $k$,
since the less frequent values are harder to identify.
In the synthetic dataset, \itergroup achieves almost full utility for up to $k=30$.
On the other hand, in some cases, as $k$ increases,
the absolute number of heavy hitters that can be identified stops increasing.
This is because the task becomes hard so that even with more guesses,
it is still hard to find .

\begin{figure*}[t]
	\subfigure[URL, \fone]{
		\label{subfig:eval_k_url_1-32}
		\includegraphics[width=0.67\columnwidth]{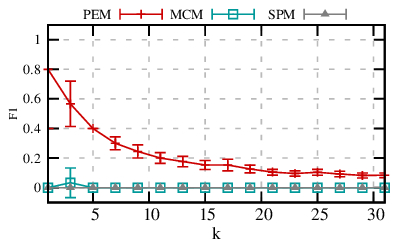}
	}
	\subfigure[AOL, \fone]{
		\label{subfig:eval_k_password_1-32}
		\includegraphics[width=0.67\columnwidth]{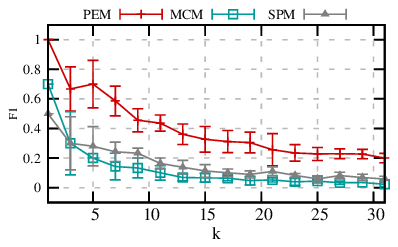}
	}
	\subfigure[Synthetic, \fone]{
		\label{subfig:eval_k_synthesize_1-32exponential}
		\includegraphics[width=0.67\columnwidth]{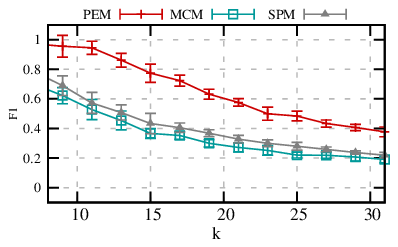}
	}

	\subfigure[URL, \ncr]{
		\label{subfig:eval_k_url_ncr_1-32}
		\includegraphics[width=0.67\columnwidth]{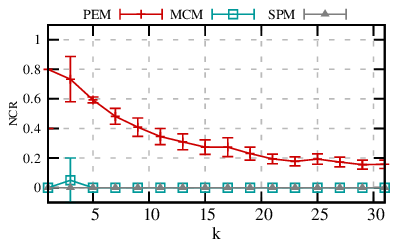}
	}
	\subfigure[AOL, \ncr]{
		\label{subfig:eval_k_query_ncr_1-32}
		\includegraphics[width=0.67\columnwidth]{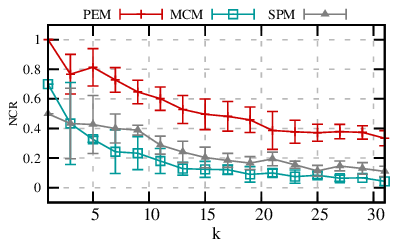}
	}
	\subfigure[Synthetic, \ncr]{
        \label{subfig:eval_k_synthesize_ncr_1-32exponential}
        \includegraphics[width=0.67\columnwidth]{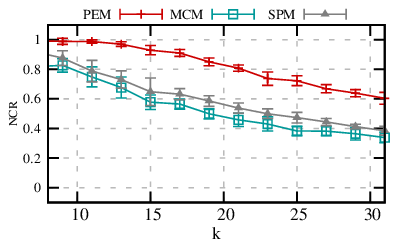}
    }
	\caption{Evaluation of the datasets, varying $k$ while fixing $\epsilon=2$. }
	\label{fig:eval_vary_k_method_1-32}
\end{figure*}



\subsubsection{Comparison of Threshold Version}
\label{sec:thres}
Both \spm and \hash use internal threshold test to find heavy hitters.
In this section, 
we modify \itergroup in order to identify heavy hitters with frequencies above a threshold $\theta$.
Note that each threshold value $\theta$ can be translated into a corresponding $k$ value.
The lower the $\theta$, the bigger the $k$ is.

Similar to the previous section, we also show results varying $\epsilon$ and $\theta$.
For brevity, we only show \fone for the Exponential dataset in Figure~\ref{fig:eval_vary_method_syn_thres}.
The results are similar in other datasets.

As can be seen from Figure~\ref{subfig:eval_eps_synthesizeexponential_thres},
when we fix $\theta=0.023$ ($0.023$ is around frequency of the $16$-the most frequent value in the dataset),
\itergroup performs better than \spm and \hash.
This advantage is most profound when $\epsilon=2$, 
where \itergroup achieves performs much better than existing methods.
The effects of fixing $\epsilon$ and varying $\theta$ are also demonstrated in Figure~\ref{subfig:eval_k_synthesize_exponential_thres} and~\ref{subfig:eval_k_synthesize_exponential_thres_e4}.

\begin{figure*}[t]
	\subfigure[\fone, vary $\epsilon$ fixing $\theta=0.023$]{
		\label{subfig:eval_eps_synthesizeexponential_thres}
		\includegraphics[width=0.67\columnwidth]{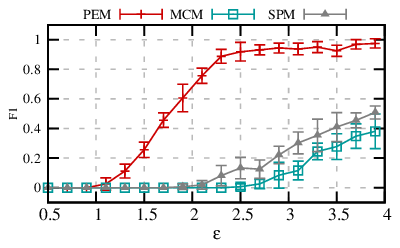}
	}
	\subfigure[\fone, vary $\theta$ fixing $\epsilon=2$]{
		\label{subfig:eval_k_synthesize_exponential_thres}
		\includegraphics[width=0.67\columnwidth]{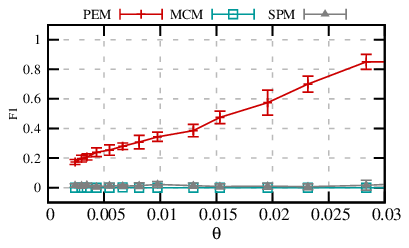}
	}
	\subfigure[\fone, vary $\theta$ fixing $\epsilon=4$]{
		\label{subfig:eval_k_synthesize_exponential_thres_e4}
		\includegraphics[width=0.67\columnwidth]{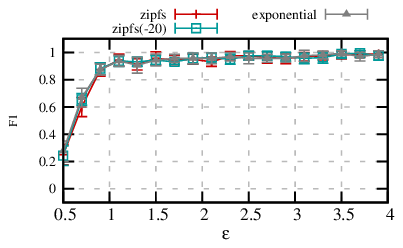}
	}
	\caption{Evaluation of the synthetic datasets, vary one of $\epsilon$ and $\theta$ while fixing the other. $m=64,n=1000000$.}
	\label{fig:eval_vary_method_syn_thres}
\end{figure*}

\subsubsection{Effect of Partitioning Users}
\label{sec:more_impro}

We further improve existing algorithms according to Proposition~\ref{thm:par_user_iden}.
Namely, instead of split privacy budget, we allocate $10\%$ of users for the final testing.
The result shown in Figure~\ref{fig:eval_improve} demonstrates the advantage of partitioning users.
Especially, when $\epsilon=1.2$,
the original \hash method achieves \fone less than $0.2$,
while the new version achieves nearly $0.8$.
For brevity, we only show \fone score on Exponential dataset, 
but the trend is similar in other settings.

\begin{figure*}[t]
    \subfigure[Effect of partitioning users.]{
        \label{fig:eval_improve}
        \includegraphics[width=0.67\columnwidth]{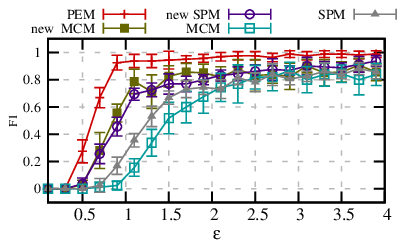}
    }
    \subfigure[Effect of $\eta$.]{
        \label{fig:eval_s}
        \includegraphics[width=0.67\columnwidth]{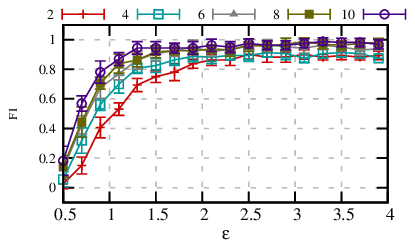}
    }
    \subfigure[Effect of distribution assumption.]{
        \label{fig:eval_dist}
        \includegraphics[width=0.67\columnwidth]{figure/eval/eval_vary_eps_synthesize_k4m64p20_dist_exponential}
    }
    \caption{Evaluation of the synthetic datasets, vary $\epsilon$. $m=64,n=1000000$. \fone is plotted.}
    \label{fig:eval_more}
\end{figure*}


\subsubsection{Effect of $\eta$}
\label{sec:more_impro_group}
In Figure~\ref{fig:eval_s}, we demonstrate the effect of segment size, i.e., $\eta$,
in \itergroup.
We fix $\eta=2, 4, 6, 8, 10$ and plot the results.
It is clear that when $\eta$ increases,
the overall utility is better.
When $\epsilon=0.9$, we see
the \fone score is $0.4$ when $\eta=2$,
and $0.8$ when $\eta=10$.
Note that there should be a limit on how large $\eta$ can be,
that is, $\eta$ is limited by the number of queries the aggregator can make.

\subsubsection{Comparison of Estimation Accuracy}
\label{sec:variance_comp}

Having demonstrated that \itergroup achieves better utility (no matter \fone or \ncr scores),
we compare the estimation accuracy.
We use the average squared error as the metric,
that is, 
\begin{align*}
\var=\frac{1}{|C_T\cap C_g|}\sum_{v\in C_T\cap C_g}\left(n_v-\tilde{n_v}\right)^2,
\end{align*}
where $n_v$ is the true count of $v$ and $\tilde{n_v}$ is its estimation by the protocol.
Note that we only account heavy hitters that are successfully identified by the protocol,
i.e., $v\in C_T\cap C_g$. 

Figure~\ref{fig:eval_var} shows comparison of estimation variance for different methods.
Observe that the \hash method has smaller variance than \rappor,
because the final testing step of \hash uses half of the $\epsilon$, 
while that of \rappor uses one third.
As a comparison, \itergroup uses only the last group,
which is one sixth of users,
and achieves similar estimation accuracy.
Note that this also complies with~\eqref{var_olh},
the estimation variance of OLH.

\begin{figure}[t]
	\centering
	\includegraphics[width=0.67\columnwidth]{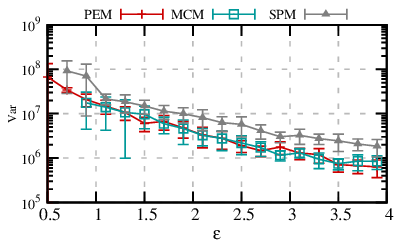}
	\caption{Evaluation of the synthetic datasets, vary $\epsilon$. $m=64,n=1000000$.}
	\label{fig:eval_var}
\end{figure}

\subsubsection{Effect of Distribution Assumption}

In the experiment, 
to mimic the blindness of the distribution,
we use a zipf's distribution to optimize \itergroup.
Note that in practice, 
it is hard to know the real distribution of the dataset.
The task of getting an accurate distribution is therefore left to the practitioners. 
Here, we argue that except in extreme cases,
the influence of a poor assumption to the final result is not much.
As we can see from Figure~\ref{fig:eval_dist}, under different assumptions,
the results are very similar.

\section{Related Work}
\label{sec:related}
Differential privacy has been the \textit{de facto} notion protecting privacy.
In the centralized settings, many DP algorithms have been proposed
(see~\cite{dpbook} for a theoretical treatment and~\cite{Li2016book} in a more practical perspective).
Recently, Uber has deployed a system enforcing DP during SQL queries~\cite{johnson2017practical},
Google also proposed several works that combine DP with machine learning~\cite{abadi2016deep, papernot2016semi}.
In the local setting, we have also seen real world deployment:
Google deployed \rap~\cite{rappor} as an extension within Chrome, and Apple~\cite{Apple} also uses similar methods to help with predictions of spelling and other things.

Of all the problems,
the basic tools in LDP are mechanisms to estimate frequencies of values.
Wang et al. compare different mechanisms using estimation variance~\cite{ldpprimitive}.
They conclude that when the domain size is small, the Generalized Random Response provides best utility, and Optimal Local Hash (\olh)/Optimal Unary Encoding (\oue)~\cite{ldpprimitive} when the domain is large.
There also exists other mechanisms with higher variance:
RAPPOR by Erlingsson et al.~\cite{rappor} and Random Matrix Projection (\blh) by Bassily and Smith~\cite{Bassily2015local}.
These protocols use ideas from earlier work~\cite{mishra2006privacy,Duchi:2013:LPS}.
Kairouz et al.~\cite{kairouz2014extremal} prove the optimal mechanisms are extreme.

The heavy hitter problem is to identify frequent values when the domain of possible values is very large, so that it is infeasible to obtain estimations for all values to identify which ones are frequent.
The problem has been studied in the centralized DP setting~\cite{chan2012differentially,mir2011pan}.
The basic idea is to use sketches to store the heavy hitters and their counts, and then publish the results with some noise.
However, in the local setting, it is not even possible to get the candidates.
One existing solution is \spm~\cite{rappor2}.
Hsu et al.~\cite{hsu2012distributed} and
Mishra et al.~\cite{mishra2006privacy} also provide efficient protocols for heavy hitters, but the error bound is proved higher than \hash, proposed by Bassily and Smith~\cite{Bassily2015local}.
In this paper, we compare with \hash~\cite{Bassily2015local} and \spm~\cite{rappor2}.

After we finish this work, we also found a simultaneous paper~\cite{bassily2017practical} by Bassily et al.
This paper proposes two methods to handle the heavy hitter problem.
The first method is similar to our \itergroup protocol except that each group of users report on one incremental bit. 
This divides users into $m$ groups.  Since this line of work is motivated by the applications where $m$ is large, 
dividing the population into $m$ groups will result in poor accuracy.  This violates one key observation we made 
in this paper: the key to improve accuracy is to reduce the number of groups. (Principle 2)

The other method is basically the \hash method with only $\sqrt{n}$ channels (instead of $n^{1.5}$ channels, as suggested in~\cite{Bassily2015local}).
In our experimental comparison with \hash, we already use around  $\sqrt{n}$ channels for \hash, and it significantly under-perform our proposed method. 
The two methods are proven to provide similar utility guarantees with similar complexities.
No experimental comparison with \spm is conducted in~\cite{bassily2017practical}.

Besides the heavy hitter problem,
there are other problems in the LDP setting that rely on mechanisms for frequency estimation.
One interesting problem is estimating frequencies of itemsets~\cite{EGS03,ESA+02}.
{Nguy{\^e}n et al.~\cite{samsung} studied how to report numerical answers.
Chen et al.~\cite{samsung_location} uses \blh to learn location from users.
Wang et al.~\cite{wang2016private} uses random response and RAPPOR together for learning weighted histogram.
Qin et al.~\cite{ccs16} estimate frequent items using RAPPOR and \blh, where each user has a set of items.
Solutions to these problems can be improved by insights gained in our paper.

Avent et al.~\cite{avent2017blender} propose a system that combines the centralized and local version of DP together and finds heavy hitters.
This work is different from ours: The dictionary of heavy hitters is constructed by a group of users who participate in the centralized version of DP.
LDP is used only to provide additional information afterwards.


\section{Conclusions}
\label{sec:conc}
In this paper, we propose LDP protocols that finds out heavy hitters in a large domain.
The utility of the protocols are thoroughly analyzed and optimized.
During analysis, we identify several design principles that can potentially serve as guidelines when solving other LDP problems.
Finally, we verify the correctness of analysis and strength of the new methods using empirical experiment on both synthetic and real-world datasets.

Current solutions rely heavily on the distribution. 
If the distribution is unfavorable, the result will be poor.
One interesting problem to explore would be to find protocols that works well for any distribution.
Another current limitation is that all the protocol proposed require the domain to be fixed. It would be interesting to find heavy hitters in an unbounded domain.
It is also an interesting direction to explore the possibility of using domain-specific knowledge to improve the protocol.

{
	\bibliographystyle{IEEEtranS}

	\bibliography{IEEEabrv,Ninghui,privacy,password}
}
\appendix

\mypara{Proof of Proposition~\ref{thm:par_user_iden}.}
\begin{proof}
	Suppose we examine the value $v_j$ in round $i$.
	We first expand $\Pall{i}{j}$ in \eqref{eqn:iden_prob} with
	$p_j = p\cdot f_j+q\cdot(1-f_j)$,
	$\sigma^2_j[i]=n[i]\cdot p_j\cdot (1-p_j), \mu_j[i]=n[i]\cdot p_j,$
	$\sigma_0^2[i]=n[i]\cdot q\cdot (1-q)$, $\mu_0[i]=n[i]\cdot q$,
	and (according to \olh) $p=\frac{1}{2}$, $q=\frac{1}{e^\epsilon+1}$ :
	\begin{align}
	&\Phi\left(\frac{\mu_j[i]+\Phi^{-1}\left(\frac{k-j}{N[i]}\right)\cdot \sigma_0[i]-\mu_0[i]}{\sigma_j[i]}\right)\nonumber\\
	=&\Phi\left(\frac{n[i]f_j(p-q)+\Phi^{-1}\left(\frac{k-j}{N[i]}\right)\cdot \sqrt{n[i]q(1-q)}}{\sqrt{n[i](q+f_j(p-q))(1-q-f_j(p-q))}}\right)\nonumber\\
	=&\Phi\left(\frac{\sqrt{n[i]}f_j(p-q)/\sqrt{q(1-q)}+\Phi^{-1}\left(\frac{k-j}{N[i]}\right)}{\sqrt{(q+f_j(p-q))(1-q-f_j(p-q))}/\sqrt{q(1-q)}}\right)\nonumber\\
	&(\mbox{using }p-q=\frac{e^\epsilon-1}{2(e^\epsilon+1)}, q(1-q)=\frac{e^\epsilon}{(1+e^\epsilon)^2})\nonumber\\
	=&\Phi\left(\frac{\sqrt{n[i]}\frac{f_j(e^\epsilon-1)}{2e^{\epsilon/2}}+\Phi^{-1}\left(\frac{k-j}{N[i]}\right)}{\sqrt{1+\left(\frac{e^\epsilon-1}{e^{\epsilon/2}}\right)^2\frac{f_j}{2}\left(1-\frac{f_j}{2}\right)}}\right)\label{eqn:full_iden}
	\end{align}
    
	Since $\Phi(\cdot)$ is monotone, when comparing~\eqref{eqn:full_iden} under different settings,
    it suffices to compare the value inside $\Phi(\cdot)$.
	Simplify this equation short notations:
	$A=\Phi^{-1}\left(\frac{k-j}{N[i]}\right)<0, B=\frac{f_j}{2}\left(1-\frac{f_j}{2}\right)
	,C=\frac{f_j\sqrt{n[i]}}{2}$, and
	$E(\epsilon)=\left(\frac{e^{\epsilon}-1}{{e^{\epsilon/2}}}\right)^2$.
	When partitioning users, 
	we have $P_1=\frac{A}{\sqrt{1+B E(\epsilon)}}+\frac{C}{\sqrt{g/E(\epsilon)+gB }};$
	when split privacy budget, we have $P_2=\frac{A}{\sqrt{1+B E(\epsilon/g)}}+\frac{C}{\sqrt{1/E(\epsilon/g)+B }}$ 
	The goal is to show $P_1>P_2$.

	It is easy to see that $E(\epsilon/g)<E(\epsilon)$, 
	and thus the subtracted term of $P_1$ is greater than $P_2$.
	As to the first terms, note that in practice, we care more about small $f_j$, 
	because values with high frequency $f_j$ can always be identified.
	Therefore, assuming $B\sim 0$, $P_1>P_2$ if $E(\epsilon/g)<E(\epsilon)/g$ (this is proven by induction on $g$ in Lemma~\ref{lm:er}).
	Numerical calculation also validate that $P_1>P_2$ in most cases.

\end{proof}

\begin{lemma}
	\label{lm:er}
	$E(\epsilon/g)<E(\epsilon)/g$.
\end{lemma}
\begin{proof}
	When $\epsilon=0$, 		
	$E(\epsilon/g)=E(\epsilon)/g$.
	When $\epsilon>0$,
	we only needs to show the derivative of the right hand is greater:
	define $z=\sqrt{e^{\epsilon/g}}$.
	\begin{align*}
	&\left(\sqrt{E(\epsilon)/g}\right)'
	=\sqrt{\frac{1}{g}}\left(z^g-\frac{1}{z^g}\right)'\\
	>&\frac{1}{g}\left(gz^{g-1}+rz^{-g-1}\right)=z^{g-1}+\frac{1}{z^{g+1}}\\
	>&1+\frac{1}{z^2}=\left(\sqrt{E(\epsilon/g)}\right)'
	\end{align*}
	The last step is proven by induction.
	
\end{proof}

\mypara{Simpler Joint Estimation for~\cite{rappor2}}

In \spm,
in order to query a value pair from $ C_i\times C_j$ where $1\leq i\neq j\leq g$,
maximal likelihood estimation (MLE) is used.

MLE updates every possible combination via a loop, which terminates only if the accumulated update amount is small.
This method is slow in practice.
Here we introduce a simpler method to recover the joint distribution:

Suppose in the group of $n$ users who report on segments $i$ and $j$, $n_a$ users have pattern $a$ in segment $i$,
$n^b$ users have pattern $b$ in segment $j$,
and $n_a^b$ users have patterns $a$ and $b$ simultaneously in segments $i$ and $j$, respectively.
The expected ``support'' this pattern pair $a,b$ receive is
\begin{align*}
E[I^b_a] = n^b_a \cdot p^2 + (n^b - n^b_a) \cdot pq + (n_a - n^b_a) \cdot pq \\
+ (n- n^b - n_a + n^b_a) \cdot q^2,
\end{align*}
where $p=\frac{1}{2}$ and $q=\frac{1}{e^\epsilon+1}$ are the parameters used in OLH.
 
Note that we already have partial estimations 
$\tilde{n_a}$ and $\tilde{n^b}$ on segment $i$ and $j$, respectively,
therefore, we can have the estimated value of $n^b_a$:
\begin{align*}
\tilde{n^b_a} &= \frac{I^b_a - (\tilde{n_a} + \tilde{n^b}) \cdot q(p-q) - n \cdot q^2}{(p - q)^2}\\
\end{align*}
Given that $\tilde{n_a}$ and $\tilde{n^b}$ are unbiased,
$\tilde{n^b_a}$ is unbiased.
In the actual implementation, we use this method.

This method can be extended to multiple (more than two) answers:
Let $V$ be a set of answers on different questions 
(e.g., $V=\{a,b\}$ in the two pattern example),

\[
\tilde{n_V} = 
\frac{I_V - \sum_{i=0}^{n-1} \left[\tilde{n^*_{V(i)}}q^{n-i}(p-q)^i\right]}{(p-q)^n}
\]
where $\tilde{n^*_{V(i)}}$ to denote the summation of joint estimation of strings specified by all size $i$ subsets of $V$, and $\tilde{n^*_\emptyset}=n$. 
\end{document}